\documentclass[a4paper,UKenglish]{lipics}

\usepackage{amsmath}
\usepackage{amssymb}
\usepackage{enumerate}

\newtheorem{krule}{Reduction Rule}
\newcommand{\avg}[1]{\gamma(#1)} 

\title{Directed Acyclic Subgraph Problem Parameterized above the Poljak-Turz\'{i}k Bound}
\author[1]{Robert Crowston}
\author[1]{Gregory Gutin}
\author[1]{Mark Jones}

\affil[1]{Royal Holloway, University of London\\
Egham TW20 0EX, UK\\ \texttt{{[robert,gutin,markj]}@cs.rhul.ac.uk}}

\Copyright[nc-nd]
          {Robert Crowston, Gregory Gutin, and Mark Jones}


\serieslogo{}
\volumeinfo
  {Billy Editor, Bill Editors}
  {2}
  {Conference title on which this volume is based on}
  {1}
  {1}
  {1}
\EventShortName{}
\DOI{10.4230/LIPIcs.xxx.yyy.p}

\begin{document}

\maketitle

\begin{abstract}
An oriented graph is a directed graph without directed 2-cycles. Poljak and Turz\'{i}k (1986) proved that every connected oriented graph $G$ on $n$ vertices and $m$ arcs contains an acyclic subgraph with at least $\frac{m}{2}+\frac{n-1}{4}$ arcs. Raman and Saurabh (2006) gave another proof of this result and left it as an open question to establish the parameterized complexity of the following problem: does $G$ have an acyclic subgraph with least $\frac{m}{2}+\frac{n-1}{4}+k$ arcs, where $k$ is the parameter? We answer this question by showing that the problem can be solved by an algorithm of runtime $(12k)!n^{O(1)}$. Thus, the problem is fixed-parameter tractable. We also prove that there is a polynomial time algorithm that either establishes that the input instance of the problem is a Yes-instance or reduces the input instance to an equivalent one of size $O(k^2)$.
\end{abstract}

\section{Introduction}

The problem of finding the maximum acyclic subgraph in a directed graph\footnote{We use standard terminology and notation on directed graphs which almost always follows \cite{BJGut}. Some less standard and this-paper-specific digraph terminology and notation is provided in the end of this section.}
is well-studied in the literature in graph theory, algorithms and their applications  alongside its dual, the feedback arc set problem, see, e.g., Chapter 15 in \cite{BJGut} and references therein. This is true, in particular, in the area of parameterized algorithmics \cite{Chen+,GutKimSZeYeo,GutYeo,RamSau}.

Each directed graph $D$ with $m$ arcs has an acyclic subgraph with at least $m/2$ arcs. To obtain such a subgraph, order the vertices $x_1,\dots , x_n$ of $D$ arbitrarily and consider two spanning subgraphs of $D$: $D'$ with arcs of the form $x_ix_j$, and $D''$  with arcs of the form $x_jx_i$, where $i<j$. One of $D'$ and $D''$ has at least $m/2$ arcs. Moreover, $m/2$ is the largest size of an acyclic subgraph in every symmetric digraph $S$ (in a symmetric digraph the existence of an arc $xy$ implies the existence of an arc $yx$).
Thus, it makes sense to consider the parameterization\footnote{We use standard terminology on parameterized algorithmics, see, e.g., \cite{DowneyFellows99,FlumGrohe06,Niedermeier06}.} above the tight bound $m/2$: decide whether a digraph $D$ contains an acyclic subgraph with at least $m/2 + k$ arcs, where $k$ is the parameter. Mahajan {\em et al.} \cite{MahRamSik} and Raman and Saurabh \cite{RamSau} asked what the complexity of this problem is. For the case of oriented graphs (i.e., directed graphs with no directed cycles of length 2), Raman and Saurabh \cite{RamSau} proved that the problem is fixed-paramter tractable. A generalization of this problem to integer-arc-weighted digraphs (where $m/2$ is replaced by the half of the total weight of $D$) was proved to be fixed-parameter tractable in \cite{GutKimSZeYeo}.

For oriented graphs, $m/2$ is no longer a tight lower bound on the maximum size of an acyclic subgraph. Poljak and Turz\'{i}k \cite{PoljakTurzik86} proved the following tight bound on the maximum size of an acyclic subgraph of a connected oriented graph $D$: $\frac{m}{2} + \frac{n-1}{4}$. To see that the bound is indeed tight consider a directed path $x_1x_2\ldots x_{2t+1}$ and add to it
arcs $x_3x_1,x_5x_3,\ldots , x_{2t+1}x_{2t-1}$. This oriented graph $H_t$ consists of $t$ directed 3-cycles and has $2t+1$ vertices and $3t$ arcs. Thus,
$\frac{m}{2} + \frac{n-1}{4}=2t$ and $2t$ is the maximum size of an acyclic subgraph of $H_t$: we have to delete an arc from every directed 3-cycle as the cycles are arc-disjoint.

Raman and Saurabh \cite{RamSau} asked to determine the parameterized complexity of the following problem: decide whether a connected oriented graph $D$ has an acyclic subgraph with at least $\frac{m}{2} + \frac{n-1}{4} + k$ arcs, where $k$ is the parameter.
Answering this question, we will prove that this problem is fixed-parameter tractable and admits a kernel with $O(k^2)$ vertices and $O(k^2)$ arcs.

Observe that we may replace $k$ by $\frac{k}{4}$ to ensure that the parameter $k$ is always integral.
Therefore, the complexity of the Raman-Saurabh problem above is equivalent to that of the following parameterized problem.

\begin{center}
\fbox{~\begin{minipage}{0.9\textwidth}
{\sc Acyclic Subgraph above Poljak-Turz\'{i}k Bound (ASAPT)}\\ \nopagebreak
  \emph{Instance:} An oriented connected graph $G$ with $n$ vertices and $m$ arcs.\\
    \nopagebreak
  \emph{Parameter:} $k$.\\ \nopagebreak
  \emph{Question:} Does $G$ contain an acyclic subgraph with at least $\frac{m}{2} + \frac{n-1}{4} + \frac{k}{4}$ arcs?
\end{minipage}~}
\end{center}


Just a few years ago, as recorded by Mahajan {\em et al.} \cite{MahRamSik}, there were only very few sporadic results on problems parameterized above or below nontrivial tight bounds. 
By now the situation has changed quite dramatically: most of the open questions in \cite{MahRamSik} on parameterized complexity of problems parameterized above or below tight bounds have been solved. In the process of solving these problems, some methods and approaches have been developed. One such method is the use of lower bounds on the maximum value of a pseudo-boolean function.
The lower bounds are obtained using either a combination of probabilistic arguments and Fourier analysis inequalities \cite{AloGutKimSzeYeo11,GutIerMniYeo,GutinKimMnichYeo,GutKimSZeYeo} or a combination of linear algebraic, algorithmic and combinatorial results and approaches \cite{CroFelGutJonRosThoYeo}. Unfortunately, this method appears to be applicable mainly to constraint satisfaction problems rather than those on graphs and, thus, development of other methods applicable to problems on graphs parameterized above or below tight bounds, is of great interest.
Recently, such a method based on linear programming was investigated in \cite{CygPilPilWoj,NarRamRamSau}.

This paper continues development of another such method, which is a combination of structural graph-theoretical and algorithmic approaches, recently introduced in \cite{CroJonMni}; in fact, this paper demonstrates that the approach of  \cite{CroJonMni} 
for designing a fixed-parameter algorithm and producing 
a polynomial-size kernel for a problem on undirected graphs parameterized above tight bound 
can be modified to achieve the same for a problem on directed graphs.

In a nutshell, the method uses both two-way reduction rules (i.e., rules reducing an instance to an equivalent one) 
and one-way reduction rules (in such a rule if the reduced instance is a {\sc Yes}-instance, then the original instance is also a {\sc Yes}-instance)
to transform the input instance to a trivial graph. If the reduction rules do not allow us to conclude that the input instance is a {\sc Yes}-instance,  then 
the input instance has a relatively ``regular'' structure that can be used to solve the problem by a fixed-parameter dynamic programming algorithm. To establish the reduction rules and to show their ``completeness'', a structural result on undirected graphs is used, such as Lemma \ref{lem:slem} in this paper or Lemma 3 in \cite{CroJonMni}.
 

While the underlying approach in both papers is the same, the proofs used are different
due to the specifics of each problem. In particular, a different set of
reduction rules is used, and the ``regular'' structure derived in our paper
is rather different from that in  \cite{CroJonMni}. The dynamic programming
algorithm and kernel proof are also completely different, other than the
fact that in both papers the proofs are based on the ``regular'' structure
of the graph. Finally, note that whilst the kernel obtained in
\cite{CroJonMni} has $O(k^5)$ vertices, we obtain a kernel with just
$O(k^2)$ vertices and $O(k^2)$ arcs.



The paper is organized as follows. In the next section,we obtain two 
basic results on oriented graphs. Two-way and one-way reduction rules 
are introduced in Sections \ref{sec:2way} and \ref{sec:1way}, respectively. Fixed-parameter
tractability of ASAPT is proved in Section \ref{sec:fpt}. Section \ref{sec:polykern} is devoted to
proving the existence of a polynomial kernel. In Section \ref{sec:openproblems}, we briefly
mention another recent paper that showed that  ASAPT is
fixed-parameter tractable. We also discuss two open questions.

\smallskip

\noindent{\bf Some Digraph Terminology and Notation.} Let $D$ be a directed graph on $n$ vertices and $m$ arcs. For a vertex $x$ in $D$, the {\em out-degree} $d^+(x)$  is the number of arcs of $D$ leaving $x$ and the {\em in-degree} $d^-(x)$ is the number of arcs of $D$ entering $x$. For a subset $S$ of vertices of $D$, let $d^+(S)$ denote the number of arcs of $D$ leaving $S$ and $d^-(S)$ the number of arcs of $D$ entering $S$.
For subsets $A$ and $B$ of vertices of $D$, let $E(A,B)$ denote the set of arcs with exactly one endpoint in each of $A$ and $B$ (in both directions).
For a set $S$ of vertices, $D[S]$ is the subgraph of $D$ induced by $S.$ When $S=\{s_1,\ldots, s_p\}$, we will write $D[s_1,\ldots, s_p]$ instead of $D[\{s_1,\ldots, s_p\}].$ The {\em underlying graph} ${\rm UN}(D)$ of $D$ is the undirected graph obtained from $D$ by replacing all arcs by edges with the same end-vertices and getting rid of one edge in each pair of parallel edges. The {\em connected components} of $D$ are connected components of
${\rm UN}(D)$; $D$ is {\em connected} if ${\rm UN}(D)$ is connected.
Vertices $x$ and $y$ of $D$ are {\em neighbors} if there is an arc between them.
The maximum number of arcs in an acyclic subgraph of $D$ will be denoted by ${\rm a}(D)$. Let $\avg{D}=\frac{m}{2}+\frac{n-c}{4}$,
where $c$ is the number of connected components of $D$. By the Poljak-Turz\'{i}k bound, we have
\begin{equation}\label{gammab}
{\rm a}(G)\ge \avg{G}
\end{equation}
for every oriented graph $G$.
A {\em tournament} is an oriented graph obtained from a complete graph by orienting its edges arbitrarily. A {\em directed $p$-cycle} is a directed cycle with $p$ arcs.

\section{Basic Results on Oriented Graphs}

In our arguments we use the following simple correspondence between acyclic digraphs and orderings of vertices in digraphs. Let $H$ be an acyclic spanning subgraph of a digraph $D$. It is well-known \cite{BJGut} and easy to see that there is an ordering $x_1,\ldots ,x_n$ of vertices of $D$ such that if $x_ix_j$ is an arc of $H$ then $i<j$. On the other hand, any ordering $x_1,\ldots ,x_n$ of vertices of a digraph $D=(V,A)$ leads to an acyclic spanning subgraph of $D$: consider the subgraph induced by $\{x_ix_j:\ x_ix_j\in A, i<j\}.$ As we study maximum-size acyclic subgraphs, we may restrict ourselves to acyclic spanning subgraphs. Thus, we may use interchangeably the notions of acyclic spanning subgraphs and vertex orderings.

There are some known lower bounds on ${\rm a}(T)$ for tournaments $T$ on $n$ vertices, see, e.g., \cite{Spen} and references therein. 
We show the following useful bound which we were unable to find in the literature.

\begin{lemma} \label{lem:tournament}
 For a tournament $T$ on $n$ vertices with $m = \binom{n}{2}$ arcs,
  we can, in polynomial time, find an acyclic subgraph with at least
  $\frac{m}{2} + \frac{3n}{4}-1=\avg{T}+\frac{2n-3}{4}$ arcs, if $n$ is even,
  or $\frac{m}{2} + \frac{3(n-1)}{4}-1=\avg{T}+\frac{2n-6}{4}$ arcs, if $n$ is odd.
\end{lemma}

\begin{proof}
We prove the lemma by induction. The claim can easily be checked for $n=1$ and $n=2$ and we may assume that $n\ge 3$.

Consider first the case when $n$ is even. 
Suppose that there exists a vertex $x$ such that $d^+(x) \ge \frac{n}{2}+1$. Consider the tournament $T'=T-x$, with $m'=m-(n-1)$ arcs and $n'=n-1$ vertices. By induction, there is an ordering on $T'$ that produces an acyclic spanning subgraph $H'$ of $T'$ such that $${\rm a}(H')\ge \frac{m'}{2} + \frac{3(n'-1)}{4}-1=
\frac{m-(n-1)}{2} + \frac{3(n-2)}{4}-1=\frac{m}{2} +\frac{3n}{4}-\frac{n}{2}-2.$$ Now add $x$ to the beginning of this ordering. This produces an acyclic spanning subgraph $H$ of $T$ such that ${\rm a}(H)\ge {\rm a}(H') +  \frac{n}{2}+1\ge \frac{m}{2} +\frac{3n}{4}-1$.

If there is a vertex $x$ such that $d^-(x) \ge \frac{n}{2}+1$, the same argument applies, but $x$ is added to the end of the ordering.

Otherwise, for every vertex $x$ of $T$, $d^+(x) \in \{ \frac{n}{2}-1, \frac{n}{2} \}$. Moreover, by considering the sum of out-degrees, exactly half the vertices have out-degree $\frac{n}{2}$. Hence, if $n\ge 4$, there are at least two vertices with out-degree $\frac{n}{2}$. Let $x$ and $y$ be two such vertices, and suppose, without loss of generality, that there is an arc from $x$ to $y$. Now consider $T'=T - \{x,y\}$ with $m'=m-(2n-3)$ edges and $n'=n-2$ vertices. By induction, there is an ordering on the vertices of $T'$ that produces an acyclic subgraph with at least $\frac{m'}{2} + \frac{3n'}{4}-1=\frac{m}{2} +\frac{3n}{4}-n-1$ arcs. Place $x$ and $y$ at the beginning of this ordering, with $x$ occurring before $y$. Then this will add all the arcs from $x$ and $y$ to the acyclic subgraph. Thus, ${\rm a}(T)\ge \frac{m}{2} +\frac{3n}{4}-n-1+n=\frac{m}{2} +\frac{3n}{4}-1$.

Now suppose that $n$ is odd. Let $x$ be any vertex in $T$, and let $T'=T - x$. By induction, there is an ordering on $T'$ that produces an acyclic subgraph with at least $\frac{m'}{2} + \frac{3n'}{4}-1$ arcs, where $n'=n-1$ is the number of vertices and $m'=m-(n-1)$ is the number of arcs in $T'$. By placing $x$ either at the beginning or end of this ordering, we may add at least $(n-1)/2$ arcs. Thus, ${\rm a}(T)\ge \frac{m-(n-1)}{2} + \frac{3(n-1)}{4} -1+ \frac{n-1}{2} = \frac{m}{2} + \frac{3(n-1)}{4}-1$.
\end{proof}

\begin{lemma}\label{lem:extend}
Let $S$ be a nonempty set of vertices of an oriented graph $G$ such that both $G-S$ and $G[S]$ are connected. If ${\rm a}(G-S)\ge  \avg{G-S}+\frac{k'}{4}$ and ${\rm a}(G[S])\ge \avg{G[S]}+\frac{k''}{4}$, then ${\rm a}(G)\ge \avg{G}+\frac{k'+k''-1}{4}+\frac{|d^+(S)-d^-(S)|}{2}$. In particular,  ${\rm a}(G)\ge \avg{G}+\frac{k'+k''-1}{4}$ if $|E(S,V(G)\setminus S)|$ is even and ${\rm a}(G)\ge \avg{G}+\frac{k'+k''+1}{4}$, if $|E(S,V(G)\setminus S)|$ is odd.
\end{lemma}

\begin{proof}
Form an acyclic subgraph on $G$ as follows. Assume without loss of generality that $d^+(S)\ge d^-(S)$. Pick the arcs leaving $S$ together with the arcs of the acyclic subgraphs in $G-S$ and $G[S]$. This forms an acyclic subgraph $H$.
Let $m=m'+m''+\bar{m}$ and $n=n'+n''$, where $G-S$ has $m'$ arcs and $n'$ vertices, $G[S]$ has $m''$ arcs and $n''$ vertices and  $\bar{m}=d^+(S)+d^-(S)$. The acyclic subgraph $H$ has at least $\avg{G-S}+\frac{k'}{4}+\avg{G[S]}+\frac{k''}{4}+\frac{\bar{m}}{2}+\frac{d^+(S)-d^-(S)}{2}=\frac{m'+m''+\bar{m}}{2} + \frac{n'-1}{4}+\frac{n''-1}{4} + \frac{k'}{4}+\frac{k''}{4}+\frac{d^+(S)-d^-(S)}{2}=\avg{G}+\frac{k'+k''-1}{4}+\frac{d^+(S)-d^-(S)}{2}$ arcs, as required.
%
 \end{proof}

\section{Two-way Reduction Rules}\label{sec:2way}

In the rest of this paper, $G$ stands for an arbitrary  connected oriented graph with $n$ vertices and $m$ arcs.
We initially apply two `two-way' reduction rules to $(G,k)$ to form a new instance $(G',k)$
such that $(G',k)$ is a {\sc Yes}-instance of {\sc ASAPT} if and only if $(G,k)$ is a {\sc Yes}-instance of {\sc ASAPT}
(i.e., the value of the parameter remains unchanged). We denote the number of vertices and arcs in $G'$ by $n'$ and $m'$, respectively.

\begin{krule}\label{rule:smallclique}
Let $x$ be a vertex and $S$ a set of two vertices such that $G[S]$ is a component of $G - x$ and $G[S\cup \{x\}]$ is a directed 3-cycle.
Then $G':=G-S.$
\end{krule}
\begin{lemma} \label{lem:smallclique}
If $(G',k)$ is an instance obtained from $(G,k)$ by an application of Rule \ref{rule:smallclique}, then $G'$ is connected, and $(G',k)$ is a {\sc Yes}-instance of {\sc ASAPT} if and only if $(G,k)$ is a {\sc Yes}-instance of {\sc ASAPT}.
\end{lemma}
\begin{proof}
Any two components of $G'-x$ will be connected by $x$ and so $G'$ is connected.
Since ${\rm a}(G')={\rm a}(G)-2$, $m'=m-3$ and $n'=n-2$, we have ${\rm a}(G)\ge \frac{m}{2}+\frac{n-1}{4}+\frac{k}{4}$ if and only if
${\rm a}(G')\ge \frac{m'}{2}+\frac{n'-1}{4}+\frac{k}{4}.$
 \end{proof}

\begin{krule}\label{rule:bridgeTriangles}
 Let $a,b,c,d,e$ be five vertices in $G$ such that $G[a,b,c]$ and $G[c,d,e]$ are directed 3-cycles, $G[a,b,c,d,e]=G[a,b,c]\cup G[c,d,e]$ and $a,e$ are the only vertices in $\{a,b,c,d,e\}$ that are adjacent to a vertex in  $G
 -\{a,b,c,d,e\}$. To obtain $G'$ from $G$, delete $b,c$ and $d$, add a new vertex $x$ and three arcs such that $G[a,x,e]$ is a directed 3-cycle.
\end{krule}
\begin{lemma}\label{lem:bridgeTriangles}
 If $(G',k)$ is an instance obtained from $(G,k)$ by an application of Rule \ref{rule:bridgeTriangles}, then $G'$ is connected, and $(G',k)$ is a {\sc Yes}-instance of {\sc ASAPT} if and only if $(G,k)$ is a {\sc Yes}-instance of {\sc ASAPT}.
\end{lemma}
\begin{proof}
Clearly, $G'$ is connected. Note that ${\rm a}(G')={\rm a}(G)-2$, $m'=m-3$ and $n'=n-2$. Thus, we have ${\rm a}(G)\ge \frac{m}{2}+\frac{n-1}{4}+\frac{k}{4}$ if and only if ${\rm a}(G')\ge \frac{m'}{2}+\frac{n'-1}{4}+\frac{k}{4}.$
 \end{proof}

\section{One-way Reduction Rules}\label{sec:1way}
Recall that $G$ stands for an arbitrary connected oriented graph with $n$ vertices and $m$ arcs.
We will apply reduction rules transforming an instance $(G,k)$  of {\sc ASAPT} into a new instance $(G',k')$, where
$G'$ is an oriented graph with $n'$ vertices and $m'$ arcs, and $k'$ is the new value of the parameter.
We will see that for the reduction rules of this section the following property will hold: if
$(G',k')$ is a {\sc Yes}-instance then $(G,k)$ is a {\sc Yes}-instance, but not necessarily vice versa.
Thus, the rules of this section are called one-way reduction rules.

\begin{krule}\label{rule:degree}
Let $x$ be a vertex such that $G-x$ is connected, and $d^+(x) \neq d^-(x)$. To obtain $(G',k')$ remove $x$ from $G$ and reduce $k$ by $2|d^+(x)-d^-(x)|-1$.
\end{krule}

\begin{lemma}\label{lem:degree}
 If $(G',k')$ is an instance reduced from $(G,k)$ by an application of Rule \ref{rule:degree}, then $G'$ is connected, and if $(G',k')$ is a {\sc Yes}-instance then $(G,k)$ is a {\sc Yes}-instance.
\end{lemma}

\begin{proof}
Let $(G',k')$ be a {\sc Yes}-instance.
Then by Lemma \ref{lem:extend} with $S=\{x\}$ and $k''=0$, ${\rm a}(G)\ge \avg{G}+\frac{k'-1}{4}+\frac{|d^+(S)-d^-(S)|}{2}=\avg{G}+\frac{k}{4}$, as required.
 \end{proof}

%

\begin{krule}\label{rule:bigclique}
 Let $S$ be a set of vertices such that $G-S$ is connected, $G[S]$ is a tournament, and $|S|\ge 4$.
 To obtain $(G',k')$, remove $S$ from $G$ and reduce $k$ by $2|S|-4$ if $S$ is even, or $2|S|-7$ if $|S|$ is odd.
\end{krule}

\begin{lemma}\label{lem:bigclique}
 If $(G',k')$ is an instance obtained from $(G,k)$ by an application of Rule \ref{rule:bigclique}, then $G'$ is connected, and if $(G',k')$ is a {\sc Yes}-instance then $(G,k)$ is a {\sc Yes}-instance.
\end{lemma}

\begin{proof}
Suppose $|S|$ is even. By Lemma \ref{lem:tournament}, ${\rm a}(G[S])\ge \avg{G[S]}+\frac{2|S|-3}{4}$. By Lemma \ref{lem:extend}, if ${\rm a}(G')\ge \avg{G'}+(k-2|S|+4)/4$, then ${\rm a}(G)\ge \avg{G}+\frac{(k-2|S|+4)+(2|S|-3)-1}{4}=\avg{G}+\frac{k}{4}$, as required.

 A similar argument applies in the case when $|S|$ is odd, except the bound from Lemma \ref{lem:tournament} is $\avg{G[S]}+\frac{2|S|-6}{4}$, and so $k'=k-(2|S|-7)$ is applied.
 \end{proof}

\begin{krule}\label{rule:triplet}
Let $S$ be a set of three vertices such that the underlying graph of $G[S]$ is isomorphic to $P_3$, and $G - S$ is connected. To obtain $(G',k')$, remove $S$ from $G$ and reduce $k$ by $1$.
\end{krule}

\begin{lemma}\label{lem:triplet}
 If $(G',k')$ is an instance obtained from $(G,k)$ by an application of Rule \ref{rule:triplet}, then $G'$ is connected, and if $(G',k')$ is a {\sc Yes}-instance then $(G,k)$ is a {\sc Yes}-instance.
\end{lemma}

\begin{proof}
Observe that ${\rm a}(G[S])=\avg{G[S]}+\frac{1}{2}$. Hence, by Lemma \ref{lem:extend}, if ${\rm a}(G')\ge \avg{G'}+(k-1)/4$, then ${\rm a}(G)\ge \avg{G}+k/4$.  \end{proof}

\section{Fixed-Parameter Tractability of {\sc ASAPT}}\label{sec:fpt}

The next lemma follows immediately from a nontrivial structural result of Crowston {\em et al.} (Lemma 3 in \cite{CroJonMni}).

\begin{lemma}\label{lem:slem}
Given any connected undirected graph $H$, at least one of the following properties holds:
\begin{description}
\item[A] There exist $v \in V(H)$ and $X \subseteq V(H)$ such that $X$ is a connected component of $H -  v$ and $X$ is a clique;
\item[B] There exist $a,b,c \in V(H)$ such that $H[\{a,b,c\}]$ is isomorphic to $P_3$ and $H -  \{a,b,c\}$ is connected;
\item[C] There exist $x,y \in V(H)$ such that $\{x,y\} \notin E(H)$, $H  - \{x,y\}$ is disconnected, and for all connected components $X$ of $H  - \{x,y\}$, except possibly one, $X \cup \{x\}$ and $X \cup \{y\}$ are cliques.
\end{description}
\end{lemma}

\begin{lemma}\label{lem:4rules}
For any connected oriented graph $G$ with at least one edge, one of Rules \ref{rule:smallclique}, \ref{rule:degree}, \ref{rule:bigclique}, \ref{rule:triplet} applies.
\end{lemma}

\begin{proof}
If there is a vertex $x\in X$ such that $G-x$ is connected and $d^+(x)\neq d^-(x)$ (we will call such a case an {\em unbalanced case}),
then Rule \ref{rule:degree} applies. Thus, assume that for each $x\in X$ such that $G-x$ is connected we have $d^+(x)=d^-(x)$.

Consider the case when property A holds. If $|X|\ge 4$, Rule \ref{rule:bigclique} applies on $S=X$.
If $|X|=3$, there has to be exactly one arc between $X$ and $v$ and $G[X]$ is a directed 3-cycle as otherwise we have an unbalanced case.
Let $x\in X$ be the endpoint of this arc in $X$. Then Rule \ref{rule:smallclique} applies with $S=X\backslash \{x\}$. If $|X|=2$, then $G[X\cup \{v\}]$ is a directed 3-cycle (as otherwise we have an unbalanced case) and so Rule \ref{rule:smallclique} applies. We cannot have $|X|=1$ as this is an unbalanced case.

If property B holds, then Rule \ref{rule:triplet} can be applied to the path $P_3$ formed by $a,b,c$ in the underlying graph of $G$.

Consider the case when property C holds. We may assume without loss of generality that the non-tournament component is adjacent to $y$.

Consider the subcase when $G - \{x,y\}$ has two connected components, $X_1$ and $X_2$, that are tournaments. Let $x_1\in X_1$, $x_2\in X_2$ and observe that the subgraph induced by $x_1,x,x_2$ forms a $P_3$ in the underlying graph of $G$
and $G - \{x_1,x,x_2\}$ is connected, and so Rule \ref{rule:triplet} applies.

Now consider the subcase when $G - \{x,y\}$ has only one connected component $X$ that is a tournament.
If $|X|\ge 3$, then $X\cup \{x\}$ is a tournament with least four vertices, and so Rule \ref{rule:bigclique} applies. If $|X|=2$, then let $X=\{a,b\}$. Observe that $a$ is adjacent to three vertices, $b,x,y$, and so we have an unbalanced case to which
Rule \ref{rule:degree} applies. Finally, $X=\{a\}$ is a singleton, then observe that $x,a,y$ form a $P_3$ in the underlying graph of $G$ and $G - \{x,a,y\}$ is connected, and so Rule \ref{rule:triplet} applies.
 \end{proof}

In this paper, we consider the one-vertex undirected graph as 2-connected. A maximal 2-connected induced subgraph of an undirected graph is called a {\em block}. An undirected graph $H$ is called a {\em forest of cliques} if each block of $H$ is a clique. A subgraph $B$ of an oriented graph $G$ is a {\em block} if ${\rm UN}(B)$ is a block in ${\rm UN}(G).$ An oriented graph $G$ is a {\em forest of cliques} if ${\rm UN}(G)$ is a forest of cliques. A connected graph $H$ that is a forest of cliques is known as a {\em tree of cliques}.

\begin{lemma}\label{lem:4props}
Given a connected oriented graph $G$ and integer $k$, we can either show that $(G,k)$ is a {\sc Yes}-instance of {\sc ASAPT}, or find a set $U$ of at most $3k$ vertices such that $G-U$ is a forest of cliques with the following properties:
\begin{enumerate}
 \item Every block in $G - U$ contains at most three vertices;
 \item Every block $X$ in $G - U$ with $|X|=3$ induces a directed 3-cycle in $G$;
 \item Every connected component in $G - U$ has at most one block $X$ with $|X|=2$ vertices;
 \item There is at most one block in $G - U$ with one vertex (i.e., there is at most one isolated vertex in $G - U$).
\end{enumerate}
\end{lemma}
\begin{proof}

Apply Rules \ref{rule:smallclique}, \ref{rule:degree}, \ref{rule:bigclique}, \ref{rule:triplet} exhaustively,
and let $U$ be the set of vertices removed by Rules \ref{rule:degree}, \ref{rule:bigclique}, and \ref{rule:triplet} (but not Rule \ref{rule:smallclique}).
If we reduce to an instance $(G'',k'')$ with $k'' \le 0$, then by Lemmas \ref{lem:smallclique}, \ref{lem:degree}, \ref{lem:bigclique} and \ref{lem:triplet}, $(G,k)$ is a {\sc Yes}-instance and we may return {\sc Yes}. Now assume that, in the completely reduced instance $(G'',k'')$, $k''>0$. We will prove that $|U|\le 3k$ and $G-U$ satisfies the four properties of the lemma.

Observe that each time $k$ is decreased by a positive integer $q$, at most $3q$ vertices are added to $U$. Thus, $|U|\le 3k$. The rest of our proof is by induction. Observe that, by Lemma \ref{lem:4rules}, for the completely reduced instance $(G'',k'')$ either $G''=\emptyset$ or $G''$ consists of a single vertex. Thus, $G''-U$ satisfies the four properties of the lemma, which forms the basis of our induction.

For the induction step, consider an instance $(G'',k'')$ obtained from the previous instance $(G',k')$ by the application of a reduction rule. By the induction hypothesis, $G''-U$ satisfies the four properties of the lemma. In the application of each of Rules \ref{rule:degree}, \ref{rule:bigclique} and \ref{rule:triplet}, the vertices deleted are added to $U$. Hence $G''-U=G'-U$ and we are done unless
$G''$ is obtained from $G'$ by an application of Rule \ref{rule:smallclique}. Recall that in Rule \ref{rule:smallclique} we delete a set $S$ such that $G[S\cup \{x\}]$ forms a directed 3-cycle. We do not add $S$ to $U$.
If $x\in G''-U$, then in $G'-U$, $S\cup \{x\}$ forms a block of size 3 that is a directed 3-cycle. If $x \notin G''-U$, then in $G'-U$, $S$ forms a new connected component with one block $S$ with $|S|=2$ vertices. Thus, $G'-U$ satisfies the four properties.
 \end{proof}

\begin{theorem}\label{thm:fpt}
There is an algorithm for {\sc ASAPT} of runtime $O((3k)!n^{O(1)}).$
\end{theorem}
\begin{proof}
We may assume that for a connected oriented graph $G$ we have the second alternative in the proof of Lemma \ref{lem:4props}, i.e., we are also given the set $U$ of at most $3k$ vertices satisfying the four properties of Lemma \ref{lem:4props}.
Consider an algorithm which generates all orderings of $U$, in time $O((3k)!)$ as $|U|\le 3k$. An ordering $u_1,u_2,...,u_{|U|}$ of $U$ means that in the acyclic subgraph of $G$ we are constructing, we keep only arcs of $G[U]$ of the form $u_iu_j$, $i<j$.
For each ordering we perform the following polynomial-time dynamic programming procedure.

For each vertex $x \in G-U$, we define a vector $(x_0,\ldots,x_{t+1})$.
Initially, set $x_i$ to be the number of vertices $u_j \in U$ with an arc from $u_j$ to $x$ if $j \le i$, or an arc from $x$ to $u_j$ if $i<j$.
Note that $x_i$ is the number of arcs between $x$ and $U$ in the acyclic subgraph under the assumption that in the ordering of the vertices of $G$, $x$ is between $u_i$ and $u_{i+1}$.

Given $v,w\in V(G-U)$ and an ordering of $U\cup \{v,w\}$, an arc $vw$ is {\em satisfiable} if there is no $u_p$ such that $v$ is after $u_p$ and $w$ is before $u_p$, for some $p\in [|U|]$. Let $T$ be a set of arcs and let $V(T)$ be the set of end-vertices of $T$. For an ordering of $U\cup V(T)$, $T$ is 
{\em satisfiable} if  each arc is satisfiable, and the set $T$ induces an acyclic subgraph.

If $G-U$ contains a block $S$ that is itself a connected component, consider $S$ and arbitrarily select a vertex $x$ of $S$. Otherwise, find a block $S$ in $G-U$ with only one vertex $x$ adjacent to other vertices in $G - U$ (such a block exists as every block including an end-vertex of a longest path in ${\rm UN}(G)-U$ is such a block). Without loss of generality, assume that $S$ has three vertices $x,y,z$ (the case $|S|=2$ can be considered similarly).

For each $i \in \{0, \dots, t+1\}$, we let $\alpha_i$ be the maximum size of a set of satisfiable arcs between $S$ and $U$ under the restriction that $x$ lies between $u_i$ and $u_{i+1}$. Observe that $\alpha_i = \max_{j,h} (x_i + y_j + z_h + \beta(i,j,h))$, where $\beta(i,j,h)$ is the maximum size of a set of satisfiable arcs in $G[S]$ under the restriction that $x$ lies between $u_i$ and $u_{i+1}$, $y$ lies between $u_j$ and $u_{j+1}$, and $z$ lies between $u_h$ and $u_{h+1}$. Now delete $S \backslash \{x\}$ from $G$, and set $x_i=\alpha_i$ for each $i$.

Continue until each component of $G - U$ consists of a single vertex. Let $x$ be such a single vertex, let $G^*$ be the original graph $G$ (i.e., given as input to our algorithm), and let $X$ be the component of $G^*-U$ containing $x$. By construction, $x_i$ is the maximum number of satisfiable arcs from arcs in $X$ and arcs between $X$ and $U$ in $G^*$, under the assumption $x$ is between $u_i$ and $u_{i+1}$. Since each vertex $x$ represents a separate component, the maximum acyclic subgraph in $G$ has $Q+\sum_{x \in V(G-U)}(\max_i x_i)$ arcs, where $Q$ is the number of arcs $u_iu_j$ in $G[U]$ such that $i<j$.

Since the dynamic programming algorithm runs in time polynomial in $n$, running the algorithm for each permutation of $U$ gives a runtime of $O((3k)!n^{O(1)})$.
\end{proof}

\section{Polynomial Kernel}\label{sec:polykern}

\begin{lemma}\label{lem:incTriangle}
Let $T$ be a directed 3-cycle, with vertices labeled 0 or 1. Then there exists an acyclic subgraph of $T$ with two arcs, such that there is no arc from a vertex labeled 1 to a vertex labeled 0.
\end{lemma}
\begin{proof}
Let $V(T)=\{a,b,c\}$ and assume that $a,b$ are labeled 0. Since $T$ is a cycle, either the arc $ac$ or $bc$ exists. This arc, together with the arc between $a$ and $b$, form the required acyclic subgraph. A similar argument holds when two vertices in $T$ are labeled $1$.
 \end{proof}

Recall that $U$ was introduced in Lemma \ref{lem:4props} as the set of vertices removed by Rules \ref{rule:degree}, \ref{rule:bigclique}, and \ref{rule:triplet}. We say that a set $\{u,a,b\}$ of vertices  is a \emph{dangerous triangle} if $u \in U$, $G[a,b]$ is a block in $G - U$, and $G[u,a,b]$ is a directed 3-cycle.

\begin{lemma} \label{lem:degreeU}
 For a vertex $u \in U$, let $t_u$ denote the number of neighbors of $u$ in $G-U$ which do not appear in a dangerous triangle containing $u$. If $t_u \ge 4k$, then we have a {\sc Yes}-instance.
\end{lemma}

\begin{proof}
Let $S$ denote the subgraph of $G-U$ consisting of all components $C$ of $G-U$ which have a neighbor of $u$. For each component $C$ of $S$, let $t_u(C)$ denote the number of neighbors of $u$ in $C$ which do not appear in a dangerous triangle containing $u$.

For each vertex $x\in G-U$, label it 0 if there exists an arc from $x$ to $u$, or 1 if there is an arc from $u$ to $x$. Recall from Lemma \ref{lem:4props} each connected component in $G-U$ has at most one block $X=\{x,y\}$ with $|X|=2$. If one vertex $x$ is labeled, assign $y$ the same label. Finally, assign label 1 to any remaining unlabeled verticies in $G-U$.

We will now construct an acyclic subgraph $H'$ of $G-U$ such that there is no arc from a vertex labeled 1 to a vertex labeled 0. We then extend this to an acyclic subgraph $H$ containing all the arcs between $u$ and $S$. 

Consider each block $X$ in $G-U$. If $|X|=3$, and $X$ is a directed 3-cycle, then by Lemma \ref{lem:incTriangle} there is an acyclic subgraph of $X$ with two arcs. Add this to $H'$.
Now suppose $|X|=2$, and let $a,b$ be the vertices of $X$ with an arc from $a$ to $b$. If $G[X \cup \{u\}]$ is a dangerous triangle, then $a$ is labeled 1 and $b$ is labeled 0. In this case we do not include the arc $ab$ in $H'$. However, $H$ will include the two arcs between $X$ and $u$, which do not count towards $t_u(C)$.
If $G[X \cup \{u\}]$ is not a dangerous triangle, then we include the arc $ab$ in the acyclic subgraph $H'$. Finally, let $H$ be the acyclic subgraph formed by adding all arcs between $u$ and $S$ to $H'$.

Observe that for each component $C$ of $S$, if $G[C \cup \{u\}]$ contains no dangerous triangle then $H$ contains at least $\avg{C}$ arcs in $G[C]$ (by the construction of $H'$) and $t_u(C)$ arcs between $C$ and $u$ (since all arcs between $S$ and $u$ are in $H$), and $\avg{C \cup \{u\}}:= \avg{G[C \cup \{u\}]} = \avg{C} + \frac{t_u(C)}{2} + \frac{1}{4}$.
So $H$ contains at least $\avg{C \cup \{u\}} + \frac{t_u(C)}{2} - \frac{1}{4}$ arcs. Since $G[C \cup \{u\}]$ contains no dangerous triangle but $C$ is adjacent to $u$, $t_u(C) \ge 1$, and so $H$ contains at least $\avg{C \cup \{u\}} + \frac{t_u(C)}{4}$ arcs.

If $G[C \cup \{u\}]$ contains a dangerous triangle then $H$ contains at least $\avg{C}-\frac{3}{4}$ arcs in $G[C]$ (this can be seen by contracting the arc in $C$ appearing in the dangerous triangle, and observing that in the resulting component $C'$, $H$ has at least $\avg{C'}$ arcs) and $t_u(C) + 2$ arcs between $C$ and $u$, and $\avg{C \cup \{u\}} = \avg{C} + \frac{t_u(C)+2}{2} + \frac{1}{4}$.
Thus, $H$ contains at least $\avg{C \cup \{u\}} + \frac{t_u(C)}{2}$ arcs.

Let $C_1, C_2, \dots, C_{q}$ be the components of $S$. Observe that $\avg{S \cup \{u\}} = \sum_{i=1}^{q}\avg{C_i \cup \{u\}}$. Then by combining the acyclic subgraphs for each $G[C_i\cup \{u\}]$, we have that $a(G[S \cup \{u\}]) \ge \sum_{i=1}^{q}(\avg{C_i \cup \{u\}} + \frac{t_u(C_i)}{4}) = \avg{S \cup \{u\}} + \frac{t_u}{4}$.


%

Finally, observe $G-S-u$ has at most $3k$ component, since each component must contain a vertex of $U$. By repeated application of Lemma \ref{lem:extend}, this implies there is an acyclic subgraph of $G$ with at least $\avg{G}+\frac{t_u-3k}{4}$ arcs.
Hence, if $t_u \ge 4k$, we have a {\sc Yes}-instance.
 \end{proof}

Using the above lemma and the fact that $|U| \le 3k$ (by Lemma \ref{lem:4props}), we have that unless $(G,k)$ is a {\sc Yes}-instance, there are at most $12k^2$ vertices in $G - U$ that are adjacent to a vertex in $U$ and do not appear in a dangerous triangle with that vertex.

\begin{lemma}\label{lem:danger}
 Let $s$ be the number of components in $G - U$ in which every neighbor $x$ of a vertex $u\in U$ appears in a dangerous triangle together with $u$.
 If $s\ge k$, we have a {\sc Yes}-instance.
\end{lemma}
\begin{proof}
By Lemma \ref{lem:4props} such a component $C_i$ contains at most one block of size 2. Since only blocks of size 2 can have vertices in dangerous triangles, only the vertices from this block in $C_i$ may be adjacent to a vertex in $U$. But since $G$ is reduced by Rule \ref{rule:smallclique}, component $C_i$ must consist of only this block. Moreover, this block must appear in at least two dangerous triangles. Let $a_i,b_i$ be the vertices of $C_i$, $i=1,\ldots ,s$ and let $C=\cup_{i=1}^s\{a_i,b_i\}$.
Let $a_ib_i$ be an arc for each $i=1,\ldots ,s$ and note that every arc of $G$ containing $a_i$ ($b_i$, respectively) is either $a_ib_i$ or is from $U$ to $a_i$ (from $b_i$ to $U$, respectively). Let $\delta_i$ be the number of dangerous triangles containing $a_i$ and $b_i$; note that $\delta_i\ge 2$.

By (\ref{gammab}), $G-C$ has an acyclic subgraph $H$ with at least $\gamma(G-C)$ arcs. Observe that we can add to $H$ all arcs entering each $a_i$ and leaving each $b_i$, $i=1,\ldots ,s$, and obtain an acyclic subgraph $H^*$ of $G$. We will prove that $H^*$ contains enough arcs to show that $(G,k)$ is a {\sc Yes}-instance. Observe that $G-C$ has at most $|U|\le 3k$ components and $G[C]$ has $2s$ vertices and $2\sum_{i=1}^s\delta_i+s$ arcs, and recall that each $\delta_i\ge 2$.
Thus, the number of arcs in $H^*$ is at least 
\begin{eqnarray*}
 \gamma(G-C)+2\sum_{i=1}^s\delta_i & \ge & 
\frac{m-2\sum_{i=1}^s\delta_i-s}{2}+\frac{n-2s-3k}{4} +
2\sum_{i=1}^s\delta_i \\
 & \ge &
 \gamma(G) + \sum_{i=1}^s\delta_i - s - \frac{3k}{4}
\ge \gamma(G) +\frac{k}{4}.
\end{eqnarray*}
\end{proof}

%
%

Let $H$ be an undirected forest of cliques, where each block contains at most three vertices. A block $B$ of $H$ is called a {\em leaf-block} if there is at most one vertex of $B$ belonging to another block of $H.$ We denote the set of leaf-blocks of $H$ by ${\cal L}(H)$. A block $B$ of $H$ is called a {\em path-block} if there is another block $B'$ of $H$ such that $B$ and $B'$ have a common vertex $c$ which belongs only to these two blocks, at most one vertex of $B$ belongs to a block other than $B'$, and at most one vertex of $B'$ belongs to a block other than $B.$ We denote the set of path-blocks which are not leaf-blocks by ${\cal P}(H)$.

\begin{lemma}\label{lem:graphbounded}
For a forest of cliques $H$, with each block of size at most three, if $l=|{\cal L}(H)|$ and $p=|{\cal P}(H)|$ then $|V(H)|\le 8l+2p$.
\end{lemma}
\begin{proof}
We prove the claim by induction on the number of blocks in $H$. The case when $H$ has only one block is trivial.
Thus, we may assume that $H$ has at least two blocks and $H$ is connected. Let $B$ be a leaf-block of $H$, and
obtain subgraph $H'$ by deleting the vertices of $B$ not belonging to another block. Note that $|V(H)|\le |V(H')|+2$.

Assume that $H'$ has a leaf-block $B'$ which is not a leaf-block in $H.$ Observe that $B'\in {\cal P}(H)$ and by induction
$|V(H)|\le 2+ 8l+2(p-1) \le 8l+2p$.

Now assume that $|{\cal L}(H')|=l-1$. Observe that removal of $B$ from $H$ may lead to a neighbour of $B$, $B'$, becoming a path-block in $H'$, together with at most two blocks neighbouring $B'$. Thus, at most three blocks may become path-blocks in $H'$. By the induction hypothesis, $|V(H')| \le 8(l-1)+2(p+3)$. Hence, $|V(H)|\le 8(l-1)+2(p+3)+2 \le 8l+2p$.
\end{proof}



\begin{theorem}
 {\sc Acyclic Subgraph above Poljak-Turz\'{i}k Bound (ASAPT)} has a kernel with $O(k^2)$ vertices and $O(k^2)$ arcs.
\end{theorem}
\begin{proof}
Consider an instance of $(G^*,k)$ of ASAPT. Apply Rules \ref{rule:smallclique} and \ref{rule:bridgeTriangles} to obtain an instance $(G,k)$ reduced by Rules \ref{rule:smallclique} and \ref{rule:bridgeTriangles}.
 
Assume that $(G,k)$ is reduced by Rules \ref{rule:smallclique} and \ref{rule:bridgeTriangles} and it is a {\sc No}-instance. 

Now we will apply all reduction rules but Rule \ref{rule:bridgeTriangles}. As a result, we will obtain the set $U$ of vertices deleted in Rules \ref{rule:degree}, \ref{rule:bigclique}, and \ref{rule:triplet}. By Lemma \ref{lem:4props}, $|U|\le 3k$ and, by Lemma \ref{lem:degreeU}, each $u \in U$ has at most $4k$ neighbors that do not appear in a dangerous triangle with $u$. By Lemma \ref{lem:danger}, there are at most $2k$ vertices in $G - U$ that appear in a dangerous triangle with every neighbor in $U$ (there are at most $k$ components, and each component has two vertices). Hence the number of neighbors in $G-U$ of vertices of $U$ is at most
$4k|U| + 2k = 12k^2 + 2k$.

Now we will adopt the terminology and notation of Lemma \ref{lem:graphbounded} (we extend it from ${\rm UN}(G-U)$ to $G-U$ as we have done earlier). Consider a leaf-block $B$. Since $G$ is reduced by Rules \ref{rule:smallclique} and \ref{rule:degree},
$B$ must contain a vertex $v$ adjacent to $U$, and furthermore, $v$ is not contained in any other block. Hence, $|{\cal L}(G-U)|\le 12k^2 + 2k$.

Next, we observe that Rule \ref{rule:bridgeTriangles} implies there do not exist two adjacent 3-vertex blocks $B = \{a,b,c\}$, $B'=\{c,d,e\}$ such that
only $a$ and $e$ belong to other blocks, unless one of $b,c,d$ has a neighbor in $U$. Observe that each connected component of $G-U$ contains at most one 2-vertex block, so there are at most $12k^2 + 2k$ 2-vertex path blocks. Each 2-vertex path block is adjacent to at most two 3-vertex path blocks. Hence, $|{\cal P}(G-U)|\le 6(12k^2 + 2k)$.
So, by Lemma \ref{lem:graphbounded},  $|V(G-U)|\le 8(12k^2 + 2k)+2\cdot 6(12k^2 + 2k)=O(k^2)$, and so $|V(G)|\le O(k^2)+3k=O(k^2)$.

Finally, we show $G$ has $O(k^2)$ arcs. There are at most $|U|^2$ arcs in $U$. Between $G-U$ and $U$ there are at most $(4k+2k)|U|$ arcs. Finally, observe that $G-U$ has at most $|V(G-U)|\le 20(12k^2 + 2k)$ blocks, and each block contains at most $3$ arcs. Hence, $|A(G)|\le |U|^2+60(12k^2 + 2k)\le 9k^2+60(12k^2 + 2k)=O(k^2)$.

Thus, either $(G,k)$ is a {\sc Yes}-instance, or $(G,k)$ forms a kernel with $O(k^2)$ vertices and $O(k^2)$ arcs.
\end{proof}

\section{Discussion}\label{sec:openproblems}

After this paper was submitted to FSTTCS 2012, we learned that Mnich {\em et al.} \cite{MPSS}  combined modified approaches of \cite{CroJonMni} and \cite{PoljakTurzik86} to
prove that a number of graph problems parameterized
above tight lower bounds are fixed-parameter tractable. In particular,
they proved that ASAPT is fixed-parameter tractable. However,
\cite{MPSS} did not obtain any results on polynomial kernels.

The algorithm of Theorem \ref{thm:fpt} has runtime $2^{O(k\log k)}n^{O(1)}$. It would be interesting to design an algorithm of runtime $2^{O(k)}n^{O(1)}$ or to prove that such an algorithm does not exist, subject to a certain complexity hypothesis, as in \cite{LokMarSau}. It would also be interesting to see whether {\sc ASAPT} admits a kernel with $O(k)$ vertices.

\end{document}